\documentclass{article}

\usepackage{arxiv}
\usepackage{thumbpdf,lmodern}
\usepackage{bm}
\usepackage{tikz, tikz-qtree}
\usetikzlibrary{positioning, matrix}

\usepackage[utf8]{inputenc}
\usepackage{xcolor}
\usepackage{amsmath,amssymb,amsthm}
\usepackage{mathabx}   
\usepackage{graphicx}
\usepackage{booktabs}
\usepackage{pdflscape}
\usepackage{cancel}
\usepackage{hyperref}
\usepackage{multirow}

\let\proglang=\textsf
\newcommand{\pkg}[1]{{\fontseries{b}\selectfont #1}}

\newtheorem{theorem}{Theorem}
\newtheorem*{remark}{Remark}
\newtheorem{corollary}{Corollary}[theorem]

\newtheorem{algorithm}{Algorithm}

\newtheorem*{example}{Example}

\title{Enhancing Forecasts Using Real-Time Data Flow and Hierarchical Forecast Reconciliation, with Applications to the Energy Sector}

\author{Lukas Neubauer, Peter Filzmoser}

\begin{document}
\maketitle

\section*{Abstract}
A novel framework for hierarchical forecast updating is presented, addressing a critical gap in the forecasting literature. By assuming a temporal hierarchy structure, the innovative approach extends hierarchical forecast reconciliation to effectively manage the challenge posed by partially observed data. This crucial extension allows, in conjunction with real-time data, to obtain updated and coherent forecasts across the entire temporal hierarchy, thereby enhancing decision-making accuracy. The framework involves updating base models in response to new data, which produces revised base forecasts. A subsequent pruning step integrates the newly available data, allowing for the application of any forecast reconciliation method to obtain fully updated reconciled forecasts. Additionally, the framework not only ensures coherence among forecasts but also improves overall accuracy throughout the hierarchy. Its inherent flexibility and interpretability enable users to perform hierarchical forecast updating concisely. The methodology is extensively demonstrated in a simulation study with various settings and comparing different data-generating processes, hierarchies, and reconciliation methods. Practical applicability is illustrated through two case studies in the energy sector — energy generation and solar power data — where the framework yields superior results compared to base models that do not incorporate new data, leading to more precise decision-making outcomes.

\section{Introduction}
In many time series forecasting scenarios, a temporal hierarchical structure exists. This hierarchy is constructed implicitly once one or multiple steps of temporal aggregation are applied to a time series of interest. In a common example, one considers both quarterly and annual observations of an economic indicator where aggregation by summing is appropriate. Naturally, the aggregated quarterly observations should equal the annual value.

When addressing the forecasting aspect of this problem, several challenges arise. Typically, the different levels of the hierarchy are modeled and forecasted independently. This approach allows for flexibility in selecting the modeling techniques. However, this leads to incoherent forecasts, i.e. the aggregated quarterly forecasts do not align with the annual forecast. To resolve this issue, hierarchical forecast reconciliation has been developed to adjust those base forecasts into coherent forecasts through a post-hoc procedure. Initially applied to cross-sectional hierarchies, the concept of hierarchical forecast reconciliation has been extended to temporal hierarchies, too. The reconciled forecasts not only become coherent but also achieve greater accuracy. Thus, hierarchical forecast reconciliation serves as a correction for model misspecification in the base models.

In analyzing temporal hierarchies of time series, an important question emerges. The fact that lower levels of the hierarchy are observable at a higher frequency should be reflected throughout the entire hierarchy. For instance, if two quarters have already been observed, this should impact both the remaining quarterly forecasts and the annual forecast. Generally, the higher frequency time series can be updated more frequently, and each update should affect higher level forecasts accordingly.

A commonly used and effective method is bottom-up aggregation. \cite{Koreisha2004UpdatingAP} show that in the setting of temporally aggregated ARMA models, bottom-up forecasts tend to provide more accurate forecasts for the aggregated time series compared to forecasts derived from the aggregated data. This advantage becomes more pronounced as more lower level data becomes available. Interestingly, this effect is only significant for short forecast horizons.

Another common concept in this context is \textit{Nowcasting}. Originally developed for weather forecasting, nowcasting has found increasing applications in economics such as forecasting GDP or inflation using higher frequency indicators. The statistical background is rather complex. By modeling the data via a state space representation, the "news" can be appropriately incorporated. Examples of such representations are factor models and vector autoregressive models. A traditional approach involves \textit{bridging equations} to relate lower frequency data to higher frequency data. Often, these models are tailored for specific hierarchies, such as aggregating daily data to monthly, making them less flexible overall. It is important to note that, in nowcasting scenarios, the goal is usually not to explain aggregated data by using the same data at a higher frequency. Instead, a time series of lower-frequency data is explained by other, higher-frequency variables. An extensive overview can be found in \cite{BANBURA2013195_nowcasting}.

When discussing the handling of time series of different frequencies, \textit{MIxed DAta Sampling} (MIDAS) models (\cite{ghysels_midas}) are frequently encountered. These regression models are able to explain a lower frequency target variable by using higher frequency regressors, making them suitable for nowcasting scenarios as well. Initially developed for economic applications, the primary aim of those models was to regress a target variable on indicators observed at a higher frequency. However, our motivation differs, as we seek to forecast a time series of interest aggregated at various frequencies and update forecasts within this hierarchy of aggregated time series. Additionally, the traditional MIDAS models were not intended to use partially observed data, whereas our approach intends to do so. It was only in subsequent advancements, such as those of \cite{Mikosch2015RealTimeFW}, that the use of real-time data and forecast updating were considered.

The post-hoc framework for hierarchical forecast reconciliation is currently designed to work only with fully observed data. For example, in the case of a quarterly and annual hierarchy, an entire year of data must be available. Recent research by \cite{DiModica_online_fcrcon} presents a recursive and adaptive approach in which an online variant is proposed. However, the authors still do not account for partially observed data and instead rely on complete rows of new data, which does not align with the issue we are facing.

Thus, this work addresses the need for a methodology to update forecasts in a hierarchical structure using partially observed data, while still allowing for flexibility in both the model and the hierarchy. Motivated by the benefits and recent research advances of hierarchical forecast reconciliation, we investigate this very flexible framework to perform forecast updating while ensuring coherence across the entire hierarchy. We propose a simple method to employ existing reconciliation techniques using partially observed data to generate updated forecasts. This approach is very flexible as the base models can be arbitrarily chosen, and the reconciliation method can be chosen based on the specific application. Once new data is available, the base forecasts can be updated, followed by a pruning step. On the pruned hierarchy, the actual reconciliation is performed, leading to updated and coherent forecasts. This procedure allows us to theoretically show an improvement in forecast accuracy as new data becomes available. Notably, the rate of improvement increases as the amount of new data grows. The potential of forecast improvements is also demonstrated on simulated data as well as real data examples from the energy sector.

The paper is structured as follows. In Section~\ref{sec:rel_lit} relevant literature is summarised. A short introduction to hierarchical forecast reconciliation is given in Section~\ref{sec:hfr_lit}, followed by its extension to temporal hierarchies (Section~\ref{sec:tagg}). 
The hierarchical forecast updating framework is given in Section~\ref{sec:alg} and is analyzed theoretically in Section~\ref{sec:meth_th}. Finally, extensive simulation studies (Section~\ref{sec:exps}) and real data applications (Section~\ref{sec:real_apps}) are presented. Here, real data from the energy domain are used, which is a field where the scenario of forecast updating is of significant interest. Concluding remarks are given in Section~\ref{sec:concl}.

\section{Related Literature}\label{sec:rel_lit}

\subsection{Hierarchical Forecast Reconciliation}\label{sec:hfr_lit}
Hierarchical forecast reconciliation seeks to generate consistent forecasts that respect the hierarchical structure of time series data. A hierarchical structure can be captured by the so-called summing matrix $S$ such that the stacked vector $\mathbf y_t$ can be expressed as $\mathbf y_t=S\mathbf b_t$, where $\mathbf b_t$ denotes the vector of the bottom level nodes of the hierarchy. The recent review of \cite{ATHANASOPOULOS2024430} presents more detailed notation as well as various examples.

A groundbreaking approach by \cite{hyndman:opt_fc_comb} framed this problem as a generalized least squares regression problem. Consider $h$-step base forecasts $\hat{\mathbf y}_{t+h|t}=\hat{\mathbf y}_h=S\hat{\mathbf b}_h$ with corresponding bottom level base forecasts $\hat{\mathbf b}_h$. The base forecasts can originate from any model approach beforehand. Without any constraint the base forecasts are not coherent, thus a secondary step is needed to obtain coherent forecasts. With this in mind, write
\begin{align}\label{eq:reg}
\hat{\mathbf y}_h = S\bm\beta_h+\bm\epsilon_h,
\end{align}
where $\bm\beta_h = \mathbb E[b_{t+h}|\mathbf y_1,\dots,\mathbf y_t]$ denotes regression coefficients, and $\bm\epsilon_h$ is the reconciliation error with covariance matrix $V_h$. The generalized linear solution yields $\hat{\bm\beta}_h=G_h\hat{\mathbf y}_h$ and reconciled forecasts $\tilde{\mathbf y}_h=SG_h\hat{\mathbf y}_h$, with $G_h = (S'V_h^{-1}S)^{-1}S'V_h^{-1}$. The matrix $G_h$ is the so-called mapping matrix, mapping base forecasts to coherent bottom level forecasts.

\cite{wick:opt_fc_recon} introduced the minimum trace (minT) estimator, recognizing that $V_h$ is unidentifiable. This approach minimizes the trace of the reconciled forecast error covariance matrix,
\begin{align}
\min_{G} \text{tr}~\text{Cov}(\mathbf y_{t+h}-\tilde{\mathbf y}_h) = \min_G \text{tr}~SG W_h G'S',
\end{align}
subject to $SGS=S$, ensuring unbiased reconciled forecasts as long as the base forecasts are unbiased, too. This results in the optimal mapping matrix $G_h = (S'W_h^{-1}S)^{-1}S'W_h^{-1}$, generalizing the regression-based solution of \eqref{eq:reg} and guaranteeing forecast coherence and unbiasedness while minimizing errors across all levels. Instead of having to use the reconciliation error covariance matrix $V_h$, this approach allows to use the covariance matrix of the base forecast errors $W_h=\text{Cov}(\mathbf{y}_{t+h} - \hat{\mathbf y}_h)$.

The minT method offers several benefits: it produces coherent and unbiased forecasts (assuming unbiased base forecasts) and improves overall performance by minimizing forecast error variance. However, these potential gains are dependent on accurate covariance matrix estimation. \cite{PANAGIOTELIS2021343} caution that for certain realizations, reconciled forecast performance may deteriorate since minT optimizes an expected loss function, especially if the covariance matrix is misspecified.

A significant challenge lies in estimating the base covariance matrix $W_h$, particularly for complex hierarchies and large forecast horizons. To address this, researchers have proposed various simplified estimators, including equal weighting, scaled reconciliation (\cite{hyndman:opt_fc_comb}), sample and shrinkage estimators for $W_h$, and structural scaling (\cite{wick:opt_fc_recon}).

\subsection{Temporal Aggregation and Temporal Forecast Reconciliation}\label{sec:tagg}
We briefly introduce the concept of temporal aggregation based on the work of \cite{ATHANASOPOULOS201760}. This notation is rather strict, so we will present it here in detail.

Let $y_t$ with $t=1,\dots,T$ be a univariate time series of interest of a certain frequency $m$.
A $k$-aggregate, where $k$ is a factor of $m$, is defined to be
\begin{align}\label{eq:tfr_agg}
    y_j^{[k]} = \sum_{t=t^\ast +(j-1)k}^{t^\ast + jk -1} y_t,\quad j=1,\dots,\lfloor T/k\rfloor,
\end{align}
where $t^\ast=T-\lfloor T/m\rfloor m+1$ is the starting point of the aggregation to ensure non-overlapping aggregates.
The resulting frequency is then $M_k=m/k$. The general aggregation scheme is defined by $k\in\{k_p,\dots,k_2,k_1\}$ with $k_p=m,k_1=1$. 
Since the index $j$ varies over the different aggregation levels, a common index is introduced. The authors set $i=1,\dots,\lfloor T/m\rfloor$ and 
\begin{align}\label{eq:tfr_index}
    y_{M_k(i-1)+z}^{[k]} = y_j^{[k]},\quad z=1,\dots,M_k,
\end{align}
such that $i$ controls the top level steps and $z$ determines the steps within each aggregation period. On the highest aggregation level, the indices align, meaning $i=j$.
That way we can write one time step of the hierarchy as the vector given by
\begin{align}\label{eq:vec_agg}
    \mathbf y_i = \left(y_i^{[m]}, \dots, {\mathbf y_i^{[{k_2}]}}',  {\mathbf y_i^{[{k_1}]}}'\right)'\quad\text{with}\quad\mathbf y_i^{[k]} = \left(y_{M_k(i-1)+1}^{[k]}, y_{M_k(i-1)+2}^{[k]}, \dots, y_{M_k i}^{[k]}\right)',
\end{align}
where $\mathbf y_i^{[k]}$
denotes the stacked entries of the time series at aggregation level $k$. 
This implies that $\mathbf y_i = S\mathbf y_i^{[1]}$, where $S$ is an appropriate summing matrix as defined in general forecast reconciliation. The corresponding forecasts $\hat{\mathbf y}_{i+h|i}$ can be obtained by stacking the single forecasts in a similar matter. Namely, the vector of stacked $h$-step ahead forecasts of level $k$ is $\hat{\mathbf y}_{i+h|i}^{[k]}$ given by
\begin{align}
    \hat{\mathbf y}_{i+h|i}^{[k]} = \left(\hat y_{M_k(i-1+h-1)+1|M_k(i-1)}^{[k]}, \hat y_{M_k(i-1+h-1)+2|M_k(i-1)}^{[k]}, \dots, \hat y_{M_k(i+h-1)|M_k(i-1)}^{[k]}\right)'.
\end{align}
Thus, this forecast actually requires forecasts of $M_k(h-1)+1,\dots,M_kh$ steps ahead.

\begin{example}
    To exemplify the notation, consider $k\in\{12,3,1\}$. This assumes $m=12$, i.e., a monthly time series, which is aggregated to quarterly data ($k=3$) and annual data ($k=12$). The corresponding frequencies are then $M_k\in\{1,4,12\}$. The stacked vector using the common index notation is
    \begin{align*}
        \mathbf y_i' = \Big(
            y_i^{[12]},
            \underbrace{
            y_{4(i-1)+1}^{[3]},
            y_{4(i-1)+2}^{[3]},
            y_{4(i-1)+3}^{[3]},
            y_{4i}^{[3]}}_{={\mathbf y_i^{[3]}}'},
            \underbrace{
            y_{12(i-1)+1}^{[1]},
            \dots,
            y_{12(i-1)+11}^{[1]},
            y_{12i}^{[1]}}_{={\mathbf y_i^{[1]}}'}
        \Big).
    \end{align*}
    The corresponding summing matrix $S\in\{0,1\}^{17\times 12}$ such that $\mathbf y_i=S\mathbf y_i^{[1]}$ is
    \begin{align*}
        \begin{pmatrix}
            1 & 1 & 1 & \hdotsfor{8} & 1 \\
            1 & 1 & 1 & 0 & \hdotsfor{7} & 0 \\
            0 & 0 & 0 & 1 & 1 & 1 & 0 & \hdotsfor{4} & 0 \\
            0 & \hdotsfor{4} & 0 & 1 & 1 & 1 & 0 & 0 & 0 \\
            0 & \hdotsfor{7} & 0 & 1 & 1 & 1 \\
            &&&&&& I_{12}
        \end{pmatrix} = 
        \begin{pmatrix}
            I_1 &\otimes &\mathbf 1_{12}' \\
             I_4 &\otimes &\mathbf 1_{3}'\\
             I_{12} &\otimes &\mathbf 1_{1}'
        \end{pmatrix},
    \end{align*}
    where $\mathbf 1_a$ denotes the vector of ones of length $a$, and $\otimes$ is the Kronecker product.
\end{example}


The relationship of $\mathbf y_i = S\mathbf y_i^{[1]}$ allows for the formulation of a similar regression problem using base forecasts. Applying the minimum trace approach yields, for a top level forecast horizon of $h\geq 1$, 
\begin{align}
    \tilde{\mathbf{ y}}_{h} = \tilde{\mathbf{ y}}_{i+h|i} = S(S'W_h^{-1}S)^{-1}S'W_h^{-1}\mathbf{\hat y}_{i+h|i},
\end{align}
where $\hat{\mathbf y}_{h} = \hat{\mathbf y}_{i+h|i}$ represents the base forecasts across all hierarchical levels stacked the way the data are, and $W_h=\text{Cov}(\mathbf y_{i+h}-\mathbf{\hat y}_h)$ is the covariance matrix of the stacked base forecast errors. Here, one usually assumes that the base forecast errors are at least jointly conditionally covariance-stationary. This formulation requires forecasts with $M_k h$-steps ahead for each aggregation level, with $M_k$ denoting the frequency of aggregation level $k$.

To address estimation challenges, researchers have proposed various simplified estimators. These include scaled reconciliation, similar to the approach of \cite{hyndman:opt_fc_comb}, and structural scaling, as suggested by \cite{wick:opt_fc_recon}. \cite{NYSTRUP2020876} introduced autocorrelation-based methods such as autocovariance scaling, Markov scaling, GLASSO for inverse cross-correlation estimation, and cross-correlation shrinkage. Further research by \cite{NYSTRUP20211127} explored dimension reduction techniques, utilizing the eigendecomposition of the cross-correlation matrix to create a filtered precision matrix - an approach particularly valuable for complex, deep hierarchies. These methods aim to enhance forecast accuracy by leveraging information across different temporal aggregation levels.

\section{Methodology}\label{sec:meth}
To this end consider the highest frequency $k=1$ or non-aggregated time index and time step $1<s<T$ where $s=m(i-1)+z$ and $1\leq z<m$ as denoted in Eq.~\eqref{eq:tfr_index}. This implies that at this time point, we have completely observed data $\mathbf y_1,\dots,\mathbf y_{i-1}$ and some partially observed data. Namely, at the bottom level these are 
\begin{align}
    y_{m(i-1)+1}^{[1]},\dots,y_{m(i-1)+z}^{[1]},
\end{align}
which leads to not yet observed, future data of
\begin{align}
y_{m(i-1)+z+1}^{[1]},\dots,y_{mi}^{[1]}.
\end{align}
Based on the number of new data on the bottom level $z$, for each level aggregation the number of new observations can be inferred. Namely, for $k\in\{m,k_{p-1},\dots,k_2,1\}$ we have $\lfloor z/k \rfloor$ new observations. A sanity check for $k=m$ yields $0$ new observations at the top level since $z<m$.

Naturally, we are interested in $\hat{\mathbf y}_i$. Following the notation in Section~\ref{sec:tagg} we consider the one-step ahead forecasts conditioned on time $i-1$. These one-step ahead forecasts, broken down into forecasts of each level of the hierarchy, following Eq.~\eqref{eq:vec_agg}, are actually 
\begin{align}
    \hat{\mathbf y}_{i|i-1} = \left(\hat y_{i|i-1}^{[m]}, \dots, (\hat{\mathbf y}_{i|i-1}^{[k_2]})',(\hat{\mathbf y}_{i|i-1}^{[{k_1}]})' \right)',
\end{align} where 
\begin{align}
    \hat{\mathbf y}_{i|i-1}^{[k]}=\left(\hat y_{M_k(i-1)+1|M_k(i-1)}^{[k]}, \hat y_{M_k(i-1)+2|M_k(i-1)}^{[k]}, \dots, \hat y_{M_k i|M_k(i-1)}^{[k]}\right)'.
\end{align}

At time $s$ we have new data available as described, and thus may update the base forecasts. Namely, by introducing the notion of $\hat{\mathbf y}_{i|i-1,z}^{[k]}=\hat{\mathbf y}_{i|z}^{[k]}$ we have
\begin{align}
    \hat{\mathbf y}_{i|z}^{[k]} &= \left(\hat y_{M_k(i-1)+1|M_k(i-1)+\lfloor z/k \rfloor}^{[k]}, \hat y_{M_k(i-1)+2|M_k(i-1)+\lfloor z/k \rfloor}^{[k]}, \dots, \hat y_{M_k i|M_k(i-1)+\lfloor z/k \rfloor}^{[k]}\right)' \nonumber \\
    &= \left( y_{M_k(i-1)+1}^{[k]}, \dots, y_{M_k(i-1)+\lfloor z/k \rfloor}^{[k]}, 
    \hat y_{M_k(i-1)+\lfloor z/k \rfloor+1|M_k(i-1)+\lfloor z/k \rfloor}^{[k]}, \dots, \hat y_{M_k i|M_k(i-1)+\lfloor z/k \rfloor}^{[k]}\right)'. \label{eq:hat_vec_upd}
\end{align}

If $\lfloor z/k\rfloor = 0$, then $\hat{\mathbf y}_{i|z}^{[k]} = \hat{\mathbf y}_{i|i-1}^{[k]}$ and no forecast can be updated on this level. Thus, the base forecasts are updated accordingly, and the first $\lfloor z/k \rfloor$ forecasts are replaced by the actually observed values as illustrated in Figure~\ref{fig:421_hier}. This figure shows an updated data row of the simple hierarchy containing quarters, half-years, as well as the entire year with $M_4=1,M_2=2,M_1=4$, and $z=2$, indicating that the first two quarters have been observed.

\begin{figure}[!ht]
    \centering
        \begin{tikzpicture}
            \tikzstyle{every node}=[font=\small]
            \matrix (m)[
                matrix of math nodes,
                nodes in empty cells,
            ] {
                \Big(\hat y_{i|i-1}^{[4]} && \node(tbi){\hat y_{2(i-1)+1|2(i-1)}^{[2]}}; && \node(tbi2){\hat y_{2i|2(i-1)}^{[2]}}; && \node(tq1){\hat y_{4(i-1)+1|4(i-1)}^{[1]}}; && \node(tq2){\hat y_{4(i-1)+2|4(i-1)}^{[1]}}; && \node(tq3){\hat y_{4(i-1)+3|4(i-1)}^{[1]}}; && \node(tq4){\hat y_{4i|4(i-1)}^{[1]}\Big)}; \\
                \\
                \\
                \Big(\hat y_{i|i-1}^{[4]} && \node(bbi){y_{2(i-1)+1}^{[2]}}; && \node(bbi2){\hat y_{2i|2(i-1)+1}^{[2]}}; && \node(bq1){y_{4(i-1)+1}^{[1]}}; && \node(bq2){y_{4(i-1)+2}^{[1]}}; && \node(bq3){\hat y_{4(i-1)+3|4(i-1)+2}^{[1]}}; && \node(bq4){\hat y_{4i|4(i-1)+2}^{[1]}\Big)}; \\
            } ;

            \begin{scope}[thick]
                \draw [->] (tbi)--(bbi);
                \draw [->] (tq1)--(bq1);
                \draw [->] (tq2)--(bq2);  
            \end{scope}
            
            \begin{scope}[dashed, blue, thick]
                \draw [->] (tbi2)--(bbi2);
                \draw [->] (tq3)--(bq3);
                \draw [->] (tq4)--(bq4);
            \end{scope}
        \end{tikzpicture}
    \caption{Illustration of an updated data row in an annual-biannual-quarterly hierarchy with two quarters of new data. A black solid arrow indicates the new data, and a blue dashed arrow shows the updated forecast. The annual forecast is not changed here.}
    \label{fig:421_hier}
\end{figure}
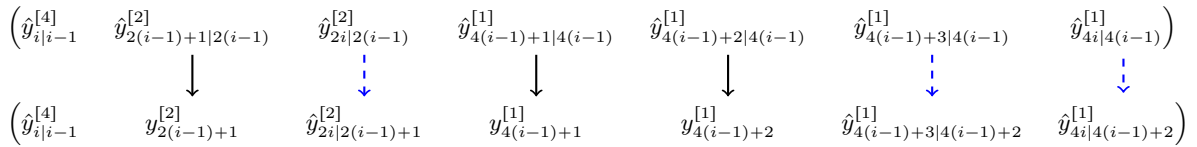

Next, the updated forecasts need to be incorporated into the entire hierarchy. For this, we use forecast reconciliation to obtain coherent forecasts across all levels of the hierarchy and utilize the new lower level information at the higher levels. Using any reconciliation method, such as the minimum trace (minT) reconciliation approach (\cite{wick:opt_fc_recon}), which utilizes the covariance of the base forecast errors, does not work easily. Two major challenges arise.

\begin{itemize}
    \item One can train the reconciliation method based on "no new data", and use this result to adjust the updated base forecasts. Without any further constraints, this may produce transformed observed values that are undesirable. Hence, more constraints are necessary, which complicates the optimization problem of determining the optimal mapping matrix. \cite{wickramasuriya2020optimal} and \cite{DIFONZO202313} illustrate the addition constraints to the reconciliation problem while \cite{ZHANG2023650} focus on the specific constraint of immutable forecasts.

    \item Alternatively, one can train the reconciliation method using a combination of observed and updated forecasts as illustrated in Figure~\ref{fig:421_hier}. This approach would produce a singular base error covariance matrix since treating observed values as forecasts would result in errors of $0$. This requires further considerations. The issue of needing additional constraints still persists.
\end{itemize}

Due to these challenges, we propose a new method that takes into account both the observed data and the updated forecasts while appropriately utilizing covariance-based reconciliation methods . To eliminate the need for constraints, we introduce a pruning step after updating the base forecasts. Using the pruned hierarchy, we then perform the forecast reconciliation step to obtain reconciled forecasts. Finally, we convert the reconciled forecasts back to the original hierarchy of interest. This approach yields updated and coherent forecasts by effectively leveraging all newly available data. More details of this procedure are given in Section~\ref{sec:alg}.

We use Figure~\ref{fig:421_hier2} to illustrate the concept of pruning. We begin by removing the observed values from each level of the hierarchy. For example, the observed values of Q1 and Q2 as well as the first biannual value are excluded from the reconciliation input. This results in a pruned hierarchy that focuses solely on reconciling the second half of the year. To achieve this, we also need to adjust the annual base forecast by subtracting the observed first biannual value. As a result, the remaining annual forecast must align with the second half-year forecast after reconciliation. This method enables us to incorporate all previous base forecasts, updated base forecasts, and any newly available data.

\begin{figure}[!ht]
    \centering
    \begin{tikzpicture}
        \tikzset{every tree node/.style={align=center,anchor=north}}
        \Tree [.\node[draw](l){Annual}; 
            [.\node(la){$\hat y_{i}^{[4]}$};
            [
            [.\node[draw, fill=green]{$\text{Biannual}_1$};
                [.\node(bi1){$\hat y_{2(i-1)+1}^{[2]}$};
                    [.\node[draw, fill=green]{$\text{Q}_1$}; \node(q1){$\hat y_{4(i-1)+1}^{[1]}$}; ] 
                    [.\node[draw, fill=green]{$\text{Q}_2$}; \node(q2){$\hat y_{4(i-1)+2}^{[1]}$}; ] 
                ] 
            ]
            [.\node[draw]{$\text{Biannual}_2$};
            [.\node(bi2){$\hat y_{2i}^{[2]}$};
                    [.\node[draw]{$\text{Q}_3$}; \node(q3){$\hat y_{4(i-1)+3}^{[1]}$}; ] 
                    [.\node[draw]{$\text{Q}_4$}; \node(q4){$\hat y_{4i}^{[1]}$}; ] 
                ] ]
        ]]]
        \begin{scope}[xshift=2.5in]
            \Tree [.\node[draw](r){Annual}; 
            [.\node(ra){$\hat y_{i}^{[4]}-y_{2(i-1)+1}^{[2]}$};
            [
            [.\node[draw]{$\text{Biannual}_2$};
            [.\node(rbi2){$\hat y_{2i|2(i-1)+1}^{[2]}$};
                    [.\node[draw]{$\text{Q}_3$}; \node(rq3){$\hat y_{4(i-1)+3|4(i-1)+2}^{[1]}$}; ] 
                    [.\node[draw]{$\text{Q}_4$}; \node(rq4) {$\hat y_{4i|4(i-1)+2}^{[1]}$}; ] 
                ] ]
        ]]]

        \end{scope}        

    \node(vec2) at (la |- bi1) {$\hat{\mathbf y}_i^{[2]}$};
    \node[left = of bi2, right = of bi1, below = 1in of vec2] (vec1) {$\hat{\mathbf y}_i^{[1]}$};

    \node(rvec1) at (ra |- vec1) {$\hat{\mathbf y}_{i|2}^{[1]}$};
    \node(rvec2) [left = of rbi2, right = of bi2] {$\hat{\mathbf y}_{i|2}^{[2]}$};
    \node(rvec4) at (rvec2 |- la) {$\hat{y}_{i|2}^{[4]}$};

    \begin{scope}[dotted]
        \draw [->] (bi1)--(vec2);
        \draw [->] (bi2)--(vec2);
        \draw [->] (q1)--(vec1);
        \draw [->] (q2)--(vec1);
        \draw [->] (q3)--(vec1);
        \draw [->] (q4)--(vec1);

        \draw [->] (rq3)--(rvec1);
        \draw [->] (rq4)--(rvec1);

        \draw [->] (rbi2)--(rvec2);
        \draw [->] (ra)--(rvec4);

    \end{scope}
    
    \path[shorten >=5pt, shorten <=5pt,bend left, black, thick] (l) edge[->] node[midway, above]{new data} (r) ;

    \end{tikzpicture}
    \caption{Visualization of an annual-biannual-quarterly temporal hierarchy in a hierarchical updating setting. Here, we assume that two quarters have already been observed, hence also the first half-year (green). Thus, we may update the base forecast of the quarterly and biannual aggregated time series, but cannot yet update the annual base forecast. The tree on the right shows the pruned hierarchy with adjusted base forecasts, which is used to perform hierarchical forecast updating.}
    \label{fig:421_hier2}
\end{figure}
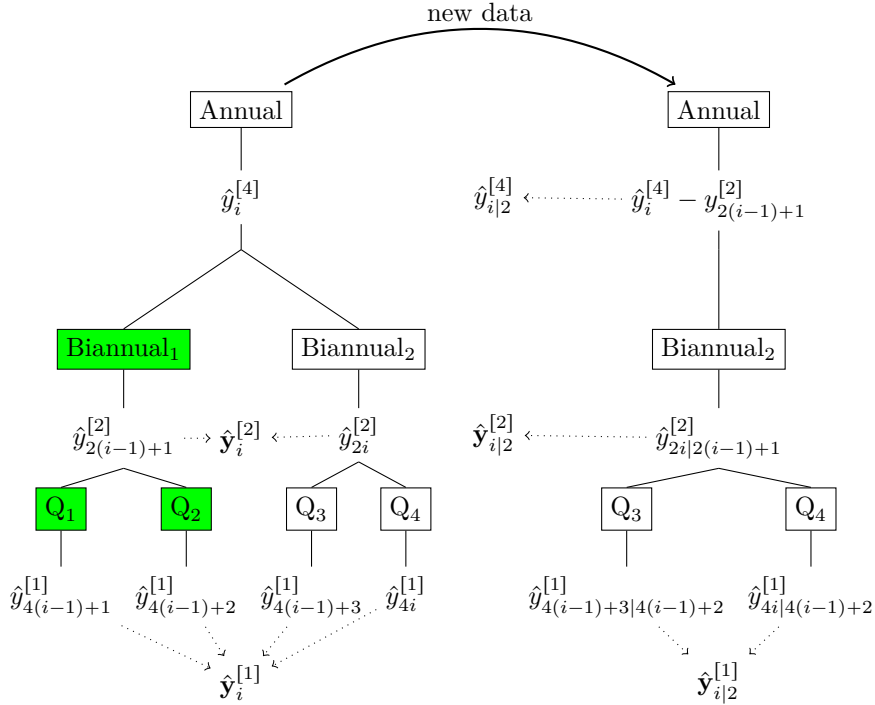

\subsection{Forecast Updating Framework}\label{sec:alg}

Consider a temporally aggregated time series $\mathbf y_i$ as defined in Eqs. \eqref{eq:tfr_agg}-\eqref{eq:vec_agg} with $k\in K=\{m, k_{p-1}, \dots, k_2,1\}$. We assume that once sufficient bottom level data are observed, the corresponding higher level data are observed, too. Let $z\in\{0, 1, \dots, m-1\}$ be the number of new observations based on the bottom level of the hierarchy. This implies higher level steps for $k$ of $\lfloor z/k \rfloor$. If $z=0$, then the following procedure performs regular hierarchical forecast reconciliation as previously denoted with fully observed rows of data. 
\begin{algorithm}[Hierarchical Forecast Updating]\label{alg:1}
Given a temporal hierarchy of time series with base forecasts and $1\leq z < m$. Updated forecasts for the entire hierarchy can be obtained as follows.
\begin{enumerate}
    \item Update base forecasts: on each level of the hierarchy obtain $m/k-\lfloor z/k\rfloor$ new forecasts based on the new data $\mathbf y_{i; 1,\dots,\lfloor z/k\rfloor}^{[k]}$, and collect the forecasts in $\hat{\mathbf y}_{i|z}^{[k]}$. 

    \item Prune hierarchy: the original hierarchy with $m_0=\sum_{k\in K} m/k$ nodes\footnote{The equality $m_0=\sum_{k\in K}k$ holds if and only if the set $\{k_{p-1},\dots,k_2\}$ contains all divisors of $m$ as in the example of $K=\{12,6,4,3,2,1\}$.} is cut to a pruned hierarchy with $m_z=\sum_{k\in K} m/k - \lfloor z/k\rfloor$ nodes. For each level $k_q$ and $\lfloor z/k_{q} \rfloor < u \leq M_{k_q}$, we compute 
        \begin{align}\label{eq:red}
            \hat{{y}}_{M_{k_q}(i-1)+u|z}^{[k_q](r)} &= 
            \hat{{y}}_{M_{k_q}(i-1)+u|z}^{[k_q]} - \sum_{p=1}^{q-1}\sum_{w=w_p}^{\lfloor z/k_p\rfloor}y_{M_{k_p}(i-1) + w}^{[k_p]}
        \end{align}
    where 
    $w_p=k_{p+1}/k_p\max(\lfloor z/k_{p+1\rfloor},u-1) + 1$. 
    This maximum is required in order to stick to the correct subhierarchy and not double-count entries.
    The observed values for $1\leq u\leq \lfloor z/k_{q} \rfloor$ are removed and not considered anymore. This results in vectors $\hat{{\mathbf y}}_{i|z}^{[k_q](r)}$ which are again stacked accordingly. In matrix notation, we can write $\hat{{\mathbf y}}_{i|z}^{(r)} = P_z\hat{{\mathbf y}}_{i|z} - R_z{\mathbf y}_i$ with the pruning matrix 
    \begin{align*}
        P_z = \text{diag}(0_{m/k-\lfloor z/k\rfloor\times\lfloor z/k\rfloor}~I_{m/k-\lfloor z/k\rfloor},k\in K),
    \end{align*}
    whereas the reduction matrix $R_z$ is given elementwise by Eq.~\eqref{eq:red}.

    \item Perform temporal hierarchical forecast reconciliation: given a mapping matrix $G_z$ and the summing matrix of the pruned hierarchy $S_z$ calculate the reconciled forecasts $\tilde{{\mathbf y}}_{i|z}^{(r)}=S_zG_z\hat{{\mathbf y}}_{i|z}^{(r)}$. The pruned summing matrix $S_z$ can actually be derived from the original hierarchy's summing matrix $S$ and is given by
    \begin{align}
        S_z = P_zS\begin{pmatrix}
            0_{z\times m-z}\\
            I_{m-z}
        \end{pmatrix}, 
    \end{align}
    
    while the mapping matrix used to perform the reconciliation is either set manually (e.g. to perform a bottom-up approach) or calculated based on data using any reconciliation approach such as minT.

    \item Add pruned nodes: calculate the reconciled forecasts for the original hierarchy by reversing Eq.~\eqref{eq:red}, i.e.
    \begin{align}\label{eq:rev_red}
        \tilde{{y}}_{M_{k_q}(i-1)+u|z}^{[k_q]}=\tilde{{y}}_{M_{k_q}(i-1)+u|z}^{[k_q](r)} + 
        \sum_{p=1}^{q-1}\sum_{w=w_p}^{\lfloor z/k_p\rfloor}y_{M_{k_p}(i-1) + w}^{[k_p]}
    \end{align}
    as well as adding back the observed values with respect to the original hierarchy, resulting in totally coherent and updated forecasts $\tilde{\mathbf y}_{i|z}$.
\end{enumerate}
\end{algorithm}

\begin{remark}  
    It is important to note that the assumed tree structure of the hierarchy is only possible if and only if the set $\{k_{p-1},\dots,k_2\}$ does not contain any co-prime pairs, i.e. there are no pairs of $k_i,k_j$ which have a common divisor greater than $1$. The example of $K=\{12,6,4,3,2,1\}$ can only be represented by two graphs, and hence the iterative pruning of Eq.~\eqref{eq:red} and reverse pruning of Eq.~\eqref{eq:rev_red} does not make sense. However, the pruning can be generalized to just using the bottom level time series. Namely,
    \begin{align}\label{eq:red_alt}
         \hat{{y}}_{M_{k_q}(i-1)+u|z}^{[k_q](r)} &= \hat{{y}}_{M_{k_q}(i-1)+u|z}^{[k_q]} -  \sum_{w=k_q(u-1) + 1}^{z}y_{m(i-1) + w}^{[1]},
    \end{align}
    where $\lfloor z/k_q\rfloor < u < M_{k_q}$ denotes the time step of aggregation level $k_q$. This representation of the pruning and reduction step is independent of the possible graph structure of the hierarchy since it uses solely bottom level information.
\end{remark}

Figure~\ref{fig:421_hier2} visualizes steps $1$ and $2$ for the simple hierarchy also used in Figure~\ref{fig:421_hier}. The base forecasts are updated wherever possible, followed by appropriately pruning the hierarchy and, therefore, regarding all newly available data. To clarify the algorithm's notation, especially of step $2$, example calculations are carried out in Appendix~\ref{app:alg_not}

It is important to note that this algorithm does not require any additional assumptions about the time series, as it operates solely on the realizations. In step 3, when we perform the forecast reconciliation to determine the optimal mapping matrix, the underlying data-generating process of the time series may be relevant, as suggested in Section~\ref{sec:meth_th}. For a theoretical analysis, additional assumptions about the time series will be necessary.

The algorithm is implemented in the \pkg{FTATS} package (\cite{FTATS-pkg}) in \proglang{R} under the function \texttt{reconcile\_forecasts\_partly} which is based on the \texttt{reconcile\_forecasts} function of the very same package. We opted to use our own implementations in contrast to already available packages such as \pkg{hts} (\cite{hts-pkg}) or \pkg{FoReco} (\cite{foreco-pkg}) to have more control over the methods and parameters as well as have more detailed outputs.

\subsection{Theoretical Analysis}\label{sec:meth_th}

Algorithm~\ref{alg:1} is analyzed to prove theoretical improvements. Theorem~\ref{th:general} gives a general statement about the improvements, while in Theorem~\ref{th:1}, we have stronger assumptions allowing us to concretely compute the properties at hand.

\begin{theorem}\label{th:general}
    Given a temporal hierarchy of time series with $K=\{m,k_{p-1},\dots,k_2,1\}$ with jointly covariance-stationary base forecast errors, Algorithm~\ref{alg:1} using the minT approach improves the base forecasts for the remaining observations, namely
    \begin{align}
        \text{tr}~\text{Cov}(P_z(\mathbf y-\tilde{\mathbf y}_z)) \leq \text{tr}~\text{Cov}(P_z(\mathbf y-\hat{\mathbf y})),
    \end{align}
    for any $0\leq z<m$. In fact, this is true for every level of the temporal hierarchy.
\end{theorem}
The proof can be found in Appendix~\ref{app:proof}. An extension is given in the following corollary.

\begin{corollary}
    As a consequence of Theorem~\ref{th:general} and some additional matrix algebra, we also have for $z_1\leq z_2$ that
    \begin{align*}
        S_{z_2}(S_{z_2}'W_{z_2}^{-1}S_{z_2})^{-1} \leq PS_{z_1}(S_{z_1}'W_{z_1}^{-1}S_{z_1})^{-1}P',
    \end{align*}
    with an appropriate pruning matrix $P$ of dimension $m_{z_2}\times m_{z_1}$. This implies that the improvements increase as more and more data are available.
\end{corollary}

With stronger assumptions about the data-generating processes, we can give a more concrete theorem. Theorem~\ref{th:1} shows that in the framework of aggregated ARIMA models, the algorithm indeed yields improved forecasts across the entire hierarchy. The assumption allows us to directly compute the covariance matrices, thus coming to a consistent result. The framework of aggregated ARIMA models tells us that this family of models is closed under aggregation. Namely, following \cite{Silvestrini2005TEMPORALAO} we have
\begin{align}\label{eq:agg_arima}
    y\sim\text{ARIMA}(p,d,q)\implies y^{[k]}\sim\text{ARIMA}(p,d,r)\quad\text{with}\quad r\leq\lfloor (p(k-1)+(d+1)(k-1)+q)/k\rfloor.
\end{align}
This is useful for the next theorem and the upcoming simulations, too.

\begin{theorem}\label{th:1}
    Given a temporal hierarchy of jointly covariance-stationary, aggregated ARIMA models with $k\in\{m,1\}$ and a bottom level $\text{AR}(1)$ model, Algorithm~\ref{alg:1} utilizing the minT approach improves overall forecast accuracy based on MSE, i.e.
    \begin{align}\label{eq:mse_ineq}
        \sum_k \text{MSE}_k(\tilde{\mathbf y}_z) \leq 
        \sum_k \text{MSE}_k(\tilde{\mathbf y}_0) \leq 
        \sum_k \text{MSE}_k(\hat{\mathbf y}_0),
    \end{align}
    for any $0 \leq z < m$ with $m>1$ denoting the bottom level frequency. $\hat{\mathbf y}_z$ denotes the updated base forecasts, and $\tilde{\mathbf y}_z$ are the updated and reconciled forecasts. $\text{MSE}_k$ denotes the mean squared error of aggregation level $k$.
\end{theorem}

The proof of Theorem~\ref{th:1} can be found in Appendix~\ref{app:proof}. 

\begin{remark}
    Without further assumptions, it is unclear how the aggregated error of the updated base forecasts $\hat{\mathbf y}_z$ relates to the series of inequalities in \eqref{eq:mse_ineq}, especially to $\sum_k \text{MSE}_k(\tilde{\mathbf y}_0)$. This is because we simultaneously have
    \begin{align*}
        \text{MSE}_1(\hat{\mathbf y}_z) &\leq \text{MSE}_1(\hat{\mathbf y}_0) = \text{MSE}_1(\tilde{\mathbf y}_0),\quad\text{and}\\
        \text{MSE}_m(\hat{\mathbf y}_z) &= \text{MSE}_m(\hat{\mathbf y}_0) \geq \text{MSE}_m(\tilde{\mathbf y}_0).
    \end{align*}
\end{remark}

The statement of Theorem~\ref{th:1} can quickly be extended to the following corollary.
\begin{corollary}\label{clry:1}
    In the setting of Theorem~\ref{th:1} and $0<z_1\leq z_2<m$ we have that
    \begin{align}
        \sum_k\text{MSE}_k(\tilde{\mathbf y}_{z_2}) \leq \sum_k\text{MSE}_k(\tilde{\mathbf y}_{z_1}).
    \end{align}
\end{corollary}
Corollary~\ref{clry:1} shows that every step of new information in the hierarchy will result in an overall lower MSE, indicating an improvement of the forecasts in total.

\section{Experiments}\label{sec:exps}

To provide a detailed illustration of Theorems~\ref{th:general} and \ref{th:1} we conduct several simulation studies as outlined below. 
Following \cite{NEUBAUER2024BU} we simulate random, stationary $\text{ARMA}(p,q)$ models of different complexity on the lowest level of the hierarchy and aggregate the resulting realizations to obtain realizations on each level of the hierarchy. Controlled parameters are the following.
\begin{itemize}
    \item Auto-regressive order $p=p_\text{bot}\in\{0,1,2\}$,
    \item Moving-average order $q=q_\text{bot}\in\{0,1,2\}$,
    \item Aggregation hierarchy $k\in\{4,1\},\{12,1\},\{12,3,1\},\{360,12,1\}$, and
    \item Base model selection automatically by AICc or using fixed and correctly specified orders (using Eq.~\eqref{eq:agg_arima}).
\end{itemize}

The top level sample size is fixed to $n_\text{top}=100$ and the innovation variance on the bottom level is set to $\sigma^2=1$. Additionally, we vary the "new data steps" $z$ based on the bottom level of the hierarchy. Specifically, for $k\in\{4,1\}$ we use $z\in\{0,1,2,3\}$ to simulate all possible new data scenarios for this type of hierarchy. For $k\in\{12,1\},\{12,3,1\}$ we have $z\in\{0,1,\dots,11\}$. The case of $k\in\{12,3,1\}$ is noteworthy because the middle level time series base forecasts may also be updated once $z=3,6,9$ as indicated in Eq.~\eqref{eq:hat_vec_upd}. As similar situation occurs for $k\in\{360,12,1\}$. In every setting, we simulate $50$ repetitions and summarize the results accordingly.

The reconciliation methods we consider in this simulation study are 
\begin{itemize}
    \item Bottom-Up: Aggregating forecasts from the lowest levelup to each higher level of the hierarchy serves as a simple baseline method.
    \item Full Cov.: In the minimum trace approach of \cite{wick:opt_fc_recon}, we estimate the complete covariance matrix of the base forecast errors and use it to calculate the optimal mapping matrix $G$.
    \item Cov. Shrinkage: A covariance shrinkage is used to shrink the estimated covariance matrix towards its diagonal matrix (\cite{wick:opt_fc_recon}). The shrinkage parameter is chosen by a cross-validation procedure.
\end{itemize}

To evaluate the results we use the root mean squared error (RMSE) summarising the reconciliation performance for each method and new data step $z$ on each level of the hierarchy of interest in a single number, and put in relation to the RMSE value of the corresponding base forecast, resulting in a relative RMSE (rRMSE) as given in Eq.~\eqref{eq:rrmse}.

\begin{align}\label{eq:rrmse}
    \text{rRMSE}^{[k]}_z(\tilde{\mathbf y}, \hat{\mathbf y}; \mathbf y) = \sqrt{\frac{\sum_i \left(y^{[k]}_i-\tilde y^{[k]}_{i|z}\right)^2}{\sum_i \left(y^{[k]}_i-\hat y^{[k]}_{i}\right)^2}}.
\end{align}

To obtain an overall error measure for an entire hierarchy, we average over each level, thus
\begin{align}
    \text{rRMSE}_z(\tilde{\mathbf y}, \hat{\mathbf y}; \mathbf y) = \frac{1}{|K|}\sum_{k\in K}\text{rRMSE}^{[k]}_z(\tilde{\mathbf y}, \hat{\mathbf y}; \mathbf y).
\end{align}

Since the reconciled forecasts include observed data points for $z>0$, using them for comparison would be unfair, as fewer errors would be considered in the error measures. Therefore, we chose to retain the errors generated when $z=0$. Let the vector $\tilde{\tilde y}_{i|z}^{[k]}$ be
\begin{align}
    \tilde{\tilde y}_{i|z}^{[k]} = \begin{pmatrix}
        \left(\tilde{\mathbf y}_{i|0}^{[k]}\right)_{1,\dots,\lfloor z/k\rfloor}\\
        \tilde{\mathbf y}_{i|z}^{(r)[k]},
    \end{pmatrix}
\end{align}
where we stack the first $\lfloor z/k\rfloor$ reconciled forecast at $z=0$ together with the reconciled forecasts at the current $z$. This approach ensures a fair comparison by incorporating the errors from the original reconciliation step. In the context of Theorem~\ref{th:general}, this adjusted error measure implies that the bottom level difference of errors where base forecasts can be updated is less significant while top level improvements remain unaffected.

The training data includes the first $n_{top}-1$ data points corresponding to the top level of the hierarchy, while the last top level observation is held out for testing. As a result, we focus on one-step ahead forecasts and assess how new data from the lower levels of the hierarchy can influence these forecasts. This approach reflects a practical application frequently encountered in real world data analysis.

We only show results for selected combinations of $p$ and $q$ since the effect of the moving average part of the model does not seem relevant for this analysis. All automatically selected models are modeled using the \texttt{auto.arima} function in the \pkg{forecast} package (\cite{forecast-pkg}) in \proglang{R}. The forecast updating algorithm is available from the \pkg{FTATS} package (\cite{FTATS-pkg}) which is also used to generate the simulation data.

The analysis is split into several parts to be able to focus on the essential aspects. We look into the changes on the top level of the hierarchy, followed by analysing the lower levels in a separate step.

\subsection{Improvements of the Top Level}
First, we take a look at the improvements of the base forecasts on the top level of the hierarchy only. Figure~\ref{fig:top_41} shows training and test relative errors for the simple hierarchy with $k\in\{4,1\}$ and automatically selected models to also include possible model misspecification. We observe that in-sample the improvements are significantly present whereby the improvements increase with the model complexity. While we use automatically selected models, there is no difference in the updating  since the data-generating models are still fully correctly specified. The picture is not as clear for the test relative improvements. There is much more uncertainty involved and clear improvements are only visible for $z\geq 2$.

\begin{figure}[!ht]
    \centering
    \includegraphics[width=1\textwidth]{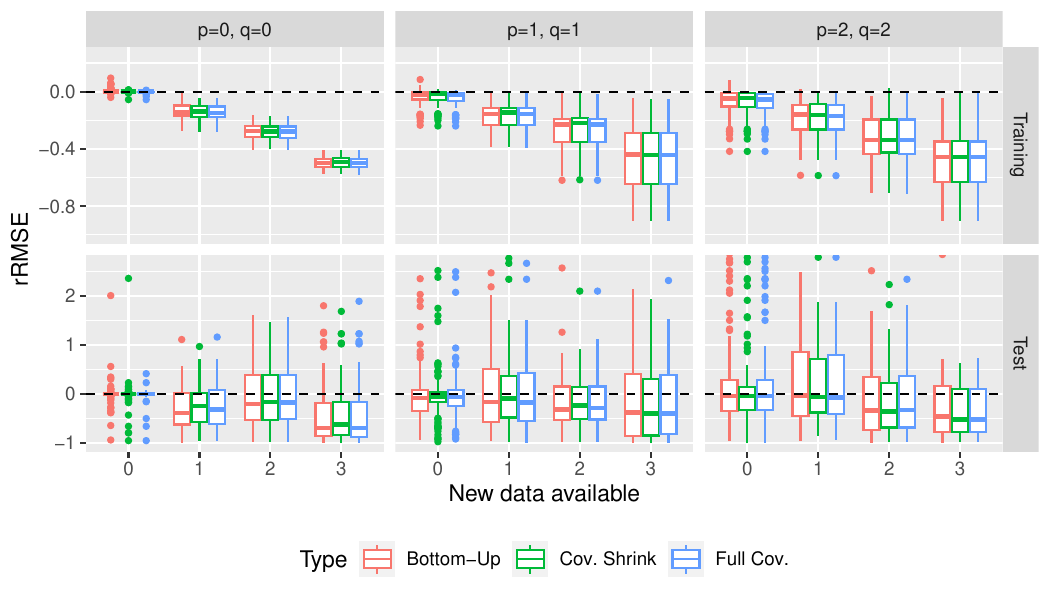}
    \caption{Training and Test rRMSE values of the top level for $k\in\{4,1\}$, and various ARMA data-generating models and automatically selected fitted models.}
    \label{fig:top_41}
\end{figure}

When fitting the correctly specified models, the conclusions remain the same. The reduction in uncertainty is mainly present through an overall tighter picture of boxplots. Thus, the updating algorithm is also capable of handling at least some model misspecification and still yields improved forecasts.

Once the hierarchy size is increased heavily, we do not show boxplots anymore. Figure~\ref{fig:top_360} shows the median relative improvements, augmented with a least squares regression line, for $k\in\{360,12,1\}$. In contrast to Figure~\ref{fig:top_41} we do not see much difference in improvements for the training errors with respect to the model complexity. This large hierarchy also presents a common problem in that the estimate of the full covariance matrix ends up being singular for most new data points, apart from rather small pruned hierarchies where $z>320$. Similarly, the test errors show median worsening until $z\approx 30$. For larger values of $z$ the improvements are positive. Overall, the improvements are increasing in $z$. As before, the methods do not yield very different results, except for the full covariance. Using this approach yields even better improvements for very large values of $z$. Interestingly, due to a much denser grid of new data points, the level of improvement can be much larger, compared to the smaller hierarchy of Figure~\ref{fig:top_41}. The inclusion of the middle level in the hierarchy does not impact the improvements of the top level at all. Results for the moderately big hierarchy with $k\in\{12,3,1\}$ look similar and are therefore not shown here.

\begin{figure}[!ht]
    \centering
    \includegraphics[width=1\textwidth]{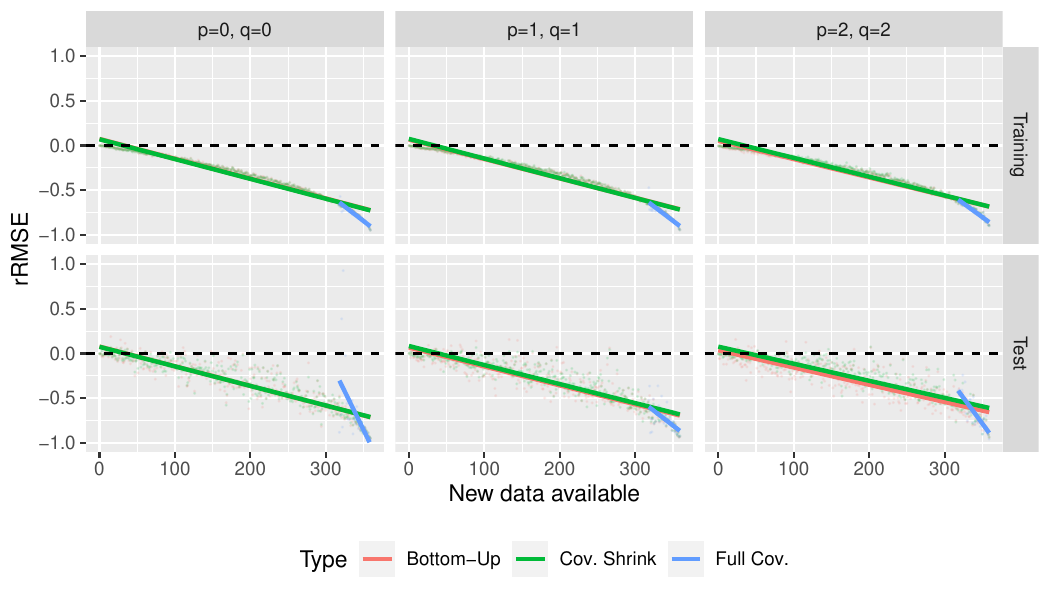}
    \caption{Training and Test rRMSE values of the top level for $k\in\{360,12,1\}$, and various ARMA data-generating models and automatically selected fitted models.}
    \label{fig:top_360}
\end{figure}

\subsection{Changes of the Bottom Levels}
Next, we analyze how the hierarchical forecast updating algorithm using forecast reconciliation affects the already updated base forecasts. Figure~\ref{fig:bot_41} shows the relative errors of the reconciliation methods as well as the updated base forecasts. Apart from the special case of a random walk, the conclusions are that the covariance-based reconciliation methods do not seem to alter the updated base forecasts for any new data time step $z$. Hence, as suggested in \cite{NEUBAUER2024BU} there are no improvements to have in the bottom level of the hierarchy, and the overall improvements are driven by the top level. In the case of $p,q=0$ we observe the strength of using the full covariance matrix estimate in the reconciliation step as this yields significantly higher improvements compared to any other approach, both on training and test data.

\begin{figure}[!ht]
    \centering
    \includegraphics[width=1\textwidth]{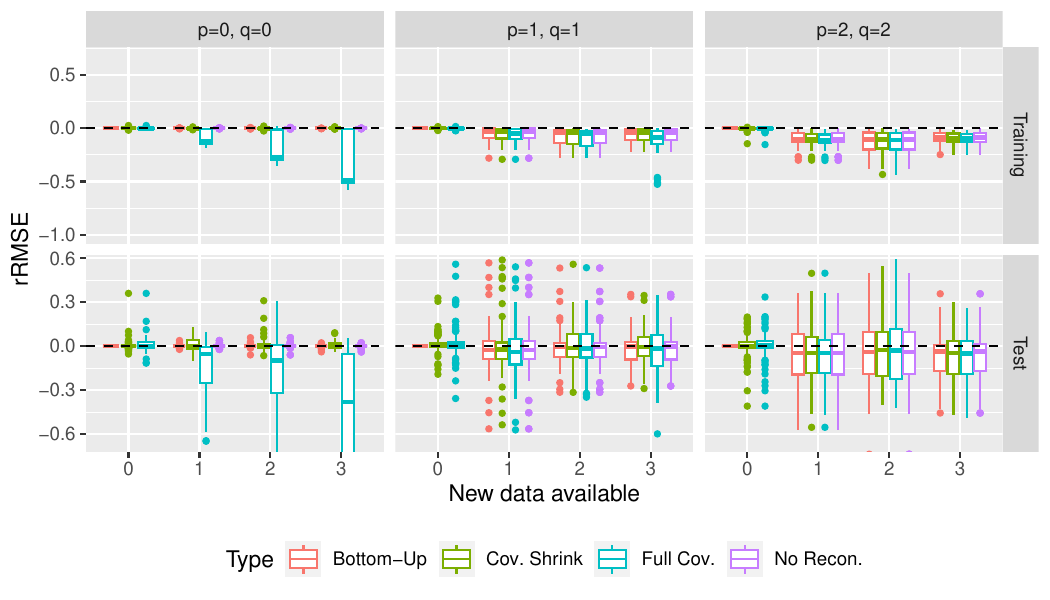}
    \caption{Training and Test rRMSE values of the bottom level for $k\in\{4,1\}$, and various ARMA data-generating models and automatically selected fitted models.}
    \label{fig:bot_41}
\end{figure}

For completeness, the very same analysis for the middle level of the hierarchy with $k\in\{12,3,1\}$ is available in Appendix~\ref{app:plots}. We observe a similar behavior as in Figure~\ref{fig:bot_41} and the random walk case. While the bottom-up approach, as well as the shrunk covariance-based minT reconciliation method, lead to no improvements of the base forecasts on the middle level, this drastically changes when considering the full covariance matrix estimate. By using the entire covariance information available, the improvements on the middle level are significantly larger on both training and test sets. Interestingly, the improvements on the top level of the hierarchy are not affected at all by this. As a consequence, this leads to better improvements on the overall level.

\subsection{Conclusions}
In this simulation study, we investigate the proposed hierarchical forecast updating algorithm using forecast reconciliation in the framework of aggregated ARIMA models. We simulated various ARMA models as well as explored a variety of hierarchies. To summarise the findings of this simulation study, we want to highlight the following points.
\begin{itemize}
    \item The model complexity and the model misspecification do not play a big factor in the performance of the updating algorithm.
    \item For simple hierarchies, all types of reconciliation lead to the same rates of improvement.
    \item The more complex the hierarchy, the higher the potential is that the reconciliation method used leads to better or worse results.
    \item As in regular forecast reconciliation, the overall improvements are driven by the top level improvements.
    \item Updated base forecasts usually remain more or less untouched by the reconciliation-based updating procedure.
\end{itemize}

\section{Real Data Applications}\label{sec:real_apps}
We demonstrate the paper's methodology on the following selected datasets suitable for hierarchical forecast updating. Here, we focus on datasets coming from the energy sector. 

Examples of recent research on applying forecast reconciliation methods in the energy sector include works of \cite{NYSTRUP2020876} and \cite{NYSTRUP20211127}, as well as \cite{DiModica_online_fcrcon}, \cite{BERGSTEINSSON2021116872}, \cite{LEPRINCE2023121510}, \cite{MOLLER2024515}. The increasing use of forecast reconciliation in this field, where accurate forecasts are crucial, justifies our focus on these datasets.

The aggregation step is done by using the \texttt{tsaggregates} function in the \pkg{thief} package (\cite{thief-pkg}) in \proglang{R}.

\subsection{Energy Generation}
Following the example of \cite{PANAGIOTELIS2023693} we consider electricity generation data of Australia measured on a daily basis for June $2019$ to May $2020$. The data is readily available from the \pkg{FTATS} package (\cite{FTATS-pkg}) in \proglang{R}. Each time series in this dataset corresponds to a certain type of generation and the corresponding amounts such as \textit{Solar} or \textit{Coal}. In total, there are $23$ bottom level time series. To obtain a setting of temporally hierarchical time series, we aggregate each daily time series into weekly (Level $2$) and then monthly data (Level $1$). To simplify, we assume that each month consists of $28$ days, i.e., we are interested in $1$,$7$, and $28$ days ahead forecasts. This yields an aggregation scheme of $k\in\{28,7,1\}$ and corresponding lengths of $336$ days, $48$ weeks and $12$ months. The assumed frequencies are then $M_k\in\{1,4,28\}$. However, we do not model any seasonalities as they did not turn out to lead to better results. As in the simulation studies, we leave out the last month to test the generalizability of the forecast updating procedure.

The base models used to model each level of each time series are automatically selected ARIMA models utilizing the \texttt{auto.arima} function of the \pkg{forecast} package in \proglang{R}. We focus on the bottom-up approach as well as using the shrunk covariance-based method for the reconciliation step of the updating algorithm. The full covariance method is excluded due to estimation problems.

\begin{figure}[!ht]
    \centering
    \includegraphics[width=1\textwidth]{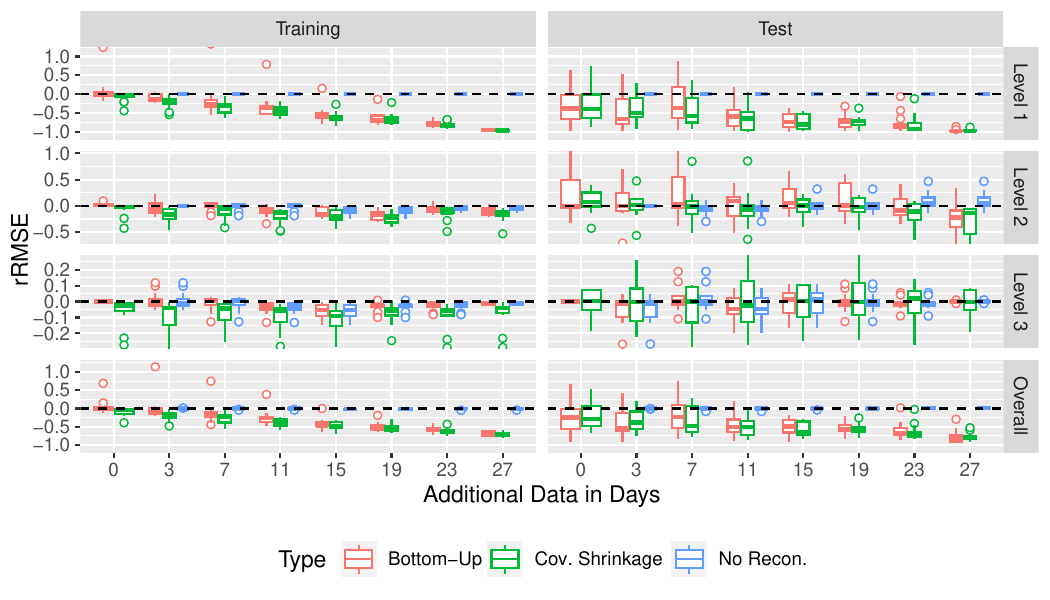}
    \caption{Training and Test rRMSE values of the Energy Generation dataset with automatically selected fitted models.}
    \label{fig:energy}
\end{figure}

Figure~\ref{fig:energy} shows training and test relative RMSE values for a selection of new data steps as well as reconciliation methods. On both the training and test set, clear trends are present. Most improvements can be achieved on the top level, whereas the change of forecasts on the lower levels remains quite small. Still, we observe that the bottom-up approach yields useful results as the difference to using the shrunk covariance in the minT reconciliation method is almost non-existing. A notable aspect of the bottom-up approach is the high variability on the test month, especially on a weekly (Level $2$) and monthly (Level $1$) basis. This is probably due to model misspecification usually observed on real data. This is again a case where a more sophisticated reconciliation method performs more reasonably. The minT approach with the full covariance matrix estimate can not be used here because of its singularity and, hence, is not shown.

\begin{figure}[!ht]
    \centering
    \includegraphics[width=1\textwidth]{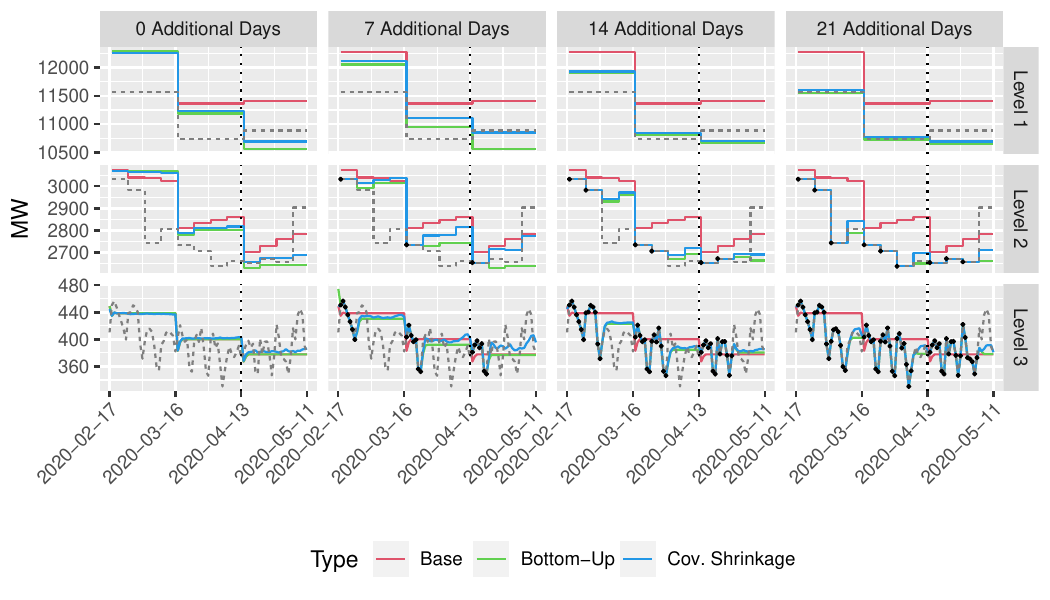}
    \caption{Results of the Non-Renewables Time Series of the Energy Generation dataset. The dashed line represents the observed data while the black dots correspond to the new data points used for forecast updating. The dotted vertical line displays the split of training and test sets.}
    \label{fig:energy_example}
\end{figure}

Actual results for the last $3$ months of the \textit{Non-Renewables} time series of the dataset are shown in Figure~\ref{fig:energy_example}. Each column corresponds to a number of new data points on a daily basis. The solid black dots show the new data points available to perform the updating algorithm. On the training set, these are used to train the models and reconciliation method accordingly (naturally, all training data points are available in total), while on the test set, they are solely used to update the models as one would do in a real application.  We observe the effect of the forecast updating procedure as, at every level, the updated base forecasts are adjusted accordingly, resulting in more accurate forecasts. The differences between the different methods are yet again very small, especially for larger values of $z$.

\subsection{Solar Power}
A second dataset used to illustrate the forecast updating procedure is the \textit{Solar Power} data analyzed by \cite{panamtash2018coherent}, which is available from \url{https://www.nrel.gov/grid/solar-power-data.html}. A processed dataset can also be found in the \proglang{R} package \pkg{FTATS} (\cite{FTATS-pkg}). In the dataset, the generation amounts of many simulated photovoltaic (PV) power plants in the United States on a $5$ minute basis are captured for the entire year of $2020$. In this analysis, we focus on $5$ states, namely Alabama, Arkansas, Arizona, California, and Colorado. In each state, we group the PV plants by their capacity, resulting in a total of $78$ bottom level time series. Further, we restrict the time series to January to keep it computationally feasible. As before, we leave out the last day of January as a test set.

The temporal hierarchy used to analyze the forecast updating procedure is defined by $k\in\{288, 12, 1\}$. This corresponds to a $5~\text{minute}$ (Level $3$) to $1~\text{hour}$ (Level $2$) to $1~\text{day}$ (Level $1$) temporal aggregation with frequencies of $M_k\in\{1,24,288\}$.

To properly model the seasonalities in the $5~\text{minute}$ and hourly time series, we use Fourier regressors of order $2$ using the \texttt{fourier} function in the \pkg{forecast} package. We do not use any regressors for the daily time series. Naturally, the PV generation numbers can be equal to $0$, especially on $5~\text{min}$ and hourly data. Hence, we log-transform (i.e. $x\mapsto\log(1+x)$ due to having $0$'s in the data) the corresponding data and apply the base models to them to handle this non-negativity. The resulting base forecasts are put back to their original scale. The base models are again automatically selected ARIMA models applying the \texttt{auto.arima} function.

To obtain non-negative reconciled forecasts, we apply the \textit{set-negative-to-zero} heuristic proposed by \cite{DIFONZO202313}. The authors use this heuristic in the solar context as well and argue little differences compared to constrained numerical optimization methods such as \cite{wickramasuriya2020optimal}. Still, using this heuristic leads to non-unbiased reconciled forecasts. The idea is to start at the bottom level, set negative forecasts to $0$, and aggregate up the differences to its neighboring higher level. This is repeated recursively until the top level is reached.

\begin{figure}[!ht]
    \centering
    \includegraphics[width=1\textwidth]{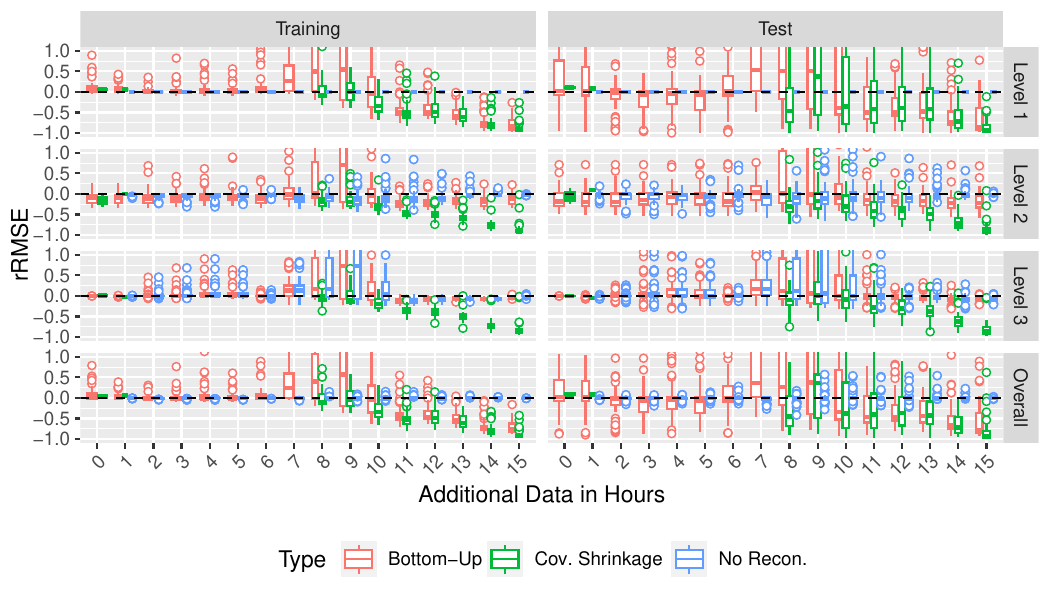}
    \caption{Training and Test rRMSE values of the Solar Power dataset with automatically selected fitted models.}
    \label{fig:pv}
\end{figure}

Figure~\ref{fig:pv} shows training and test relative errors of the bottom-up approach, the shrunk covariance-based minT reconciliation, as well as using updated forecasts and performing the reconciliation step. We have to leave out using the full covariance matrix estimate as previously observed due to rank issues.

We observe multiple interesting aspects here. First, up until $6$ hours of new data, i.e., $00:01$ up to $06:00$, the updating procedure results in minor changes in forecasts. This is due to the corresponding time series being equal to $0$; at night time the PV plants are unable to produce any power. As a result, the new data do not provide useful information. Interestingly, with some exceptions the shrunk covariance-based method fails completely in this time range. An detailed analysis reveals that not only the complete covariance matrix estimate is singular but also its diagonal matrix of variances is singular. Consequently, any linear combination as performed when shrinking is also singular, preventing the reconciliation transformation from being executed. A major reason for these near $0$ variances, especially for longer forecast horizons on the lower levels of the hierarchy, is the constant forecasts. Therefore, a reconciliation step is needed where the covariance matrix is estimated in a more reasonable manner.

Once useful new data is available, we observe clear differences between the reconciliation approaches. While the bottom-up-based updated forecasts result in high relative errors indicating worsening of the base forecasts, the covariance-based method performs significantly better. This is because the updated base forecasts overshoot once new useful data is available. The more sophisticated covariance-based method effectively manages this overshoot. However, once most of the new useful data has been observed, both methods ultimately yield similar results.

Overall, applying the proposed hierarchical forecast updating algorithm using forecast reconciliation leads to significant improvements across all levels and, thus, also on an overall scale.

\begin{figure}[!ht]
    \centering
    \includegraphics[width=1\textwidth]{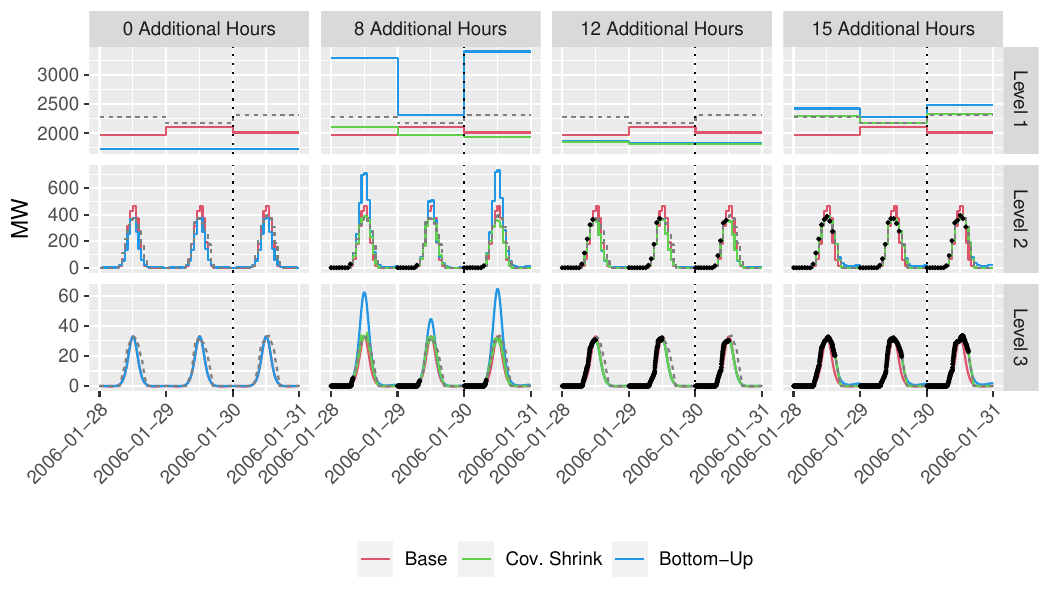}
    \caption{Results of the Californian $8\text{MW}$ Time Series of the Solar Power dataset. The dashed line represents the observed data while the black dots correspond to the new data points used for forecast updating. The dotted vertical line displays the split of training and test sets.}
    \label{fig:pv_example}
\end{figure}

Figure~\ref{fig:pv_example} shows the last $3$ days of the \textit{Californian $8\text{MW}$} time series with multiple new data steps in hours. In this example the overshoot of the updated base forecasts and thus also of the bottom-up updated forecasts is very present, especially for $8$ hours of additional data. The updating algorithm using a non-bottom-up approach acts as a way to correct the overshooting of the updated base forecasts on lower levels and, simultaneously, improve forecast accuracy on the top level, as also seen in the overview plot of Figure~\ref{fig:pv}.

\section{Conclusions}\label{sec:concl}
In this paper, a novel framework for hierarchical forecast updating is presented. Assuming a temporal hierarchical structure of time series such as monthly, quarterly, and annual data, the procedure makes use of hierarchical forecast reconciliation, a trending topic within the forecasting community. This approach allows for the integration of new data updates not only into the specific time series but also throughout the entire hierarchy of time series, thereby enhancing the overall forecasts.

To overcome possible issues of having singular covariance matrices or having to heavily constrain the optimization problem on which forecast reconciliation is built, a hierarchy pruning step is proposed. On the pruned hierarchy the reconciliation transformation is used to obtain updated and coherent forecasts for the entire hierarchy. 

The presented framework is very flexible allowing almost any family of base models, and supports many types of common covariance matrix estimators used in temporal hierarchical forecast reconciliation. By having the possibility to supply any covariance matrix estimate, any type of covariance-based forecast reconciliation approach may be used.

In our theoretical analysis, we show that the algorithm leads to improvements in forecast accuracy based on the mean squared error. This improvement is further validated through extensive simulation studies. Additionally, we apply the proposed algorithm to two real data sets from the energy sector. Both data sets confirm the viability and usefulness of the updating algorithm, providing significant support to the practical forecasting community.

Due to the high availability of data today, with much of it accessible at a high frequency, constructing a temporal hierarchy is now more cost-effective than ever. This is true even when the higher frequency time series are not of primary interest. As demonstrated, utilizing the temporal hierarchy along with an updating framework leads to significant improvements at the top level, specifically for the time series and frequency that are of interest. Therefore, the novel approach presented offers a clear and effective method to enhance time series forecasts.

There are still some aspects to consider. The algorithm should not be limited to temporal hierarchies alone. It can also be extended to include cross-sectional or even cross-temporal hierarchies. This extension is particularly useful in many cross-sectional applications. For instance, when dealing with different regions, one region may report data earlier than another. The updating algorithm aids in incorporating this partially observed data into the entire hierarchy accordingly.

It is important to investigate how the reconciliation transformation evolves over time. This is especially relevant for larger hierarchies, as there may be more efficient ways to perform the reconciliation step between different time steps.

One possible extension of our work is the choice of aggregation function. In this paper, we mainly focused on using the sum; however, in practical applications, other forms of aggregation may be more appropriate, such as stock aggregation or the median. To implement these alternatives, adjustments would need to be made to both the pruning and reconciliation steps.

To ensure robustness, it's important to investigate the impact of any unusual new data. This may significantly alter the initial forecasts, which can subsequently lead to substantial changes in the reconciled forecasts at higher levels of the hierarchy. The updating framework should be designed to manage these situations effectively and minimize sensitivity to such changes. For this, the use of probabilistic forecast reconciliation in order to measure uncertainty might be helpful.

\section*{Computational details}
The simulations and data examples were carried out in \textbf{R} 4.3.0. The corresponding source code of this paper, as well as the implementation of the algorithm, is available from GitHub at \url{https://github.com/neubluk/FTATS}.


\section*{Acknowledgments and Disclosure of Funding}
We acknowledge support from the Austrian Research Promotion Agency (FFG), Basisprogramm project “Meal Demand Forecast” and Schrankerl GmbH for the cooperation and access to their data. We further acknowledge funding from the Austrian Science
Fund (FWF) for the project “High-dimensional statistical learning: New methods to
advance economic and sustainability policies” (ZK 35), jointly carried out by WU
Vienna University of Economics and Business, Paris Lodron University Salzburg, TU
Wien, and the Austrian Institute of Economic Research (WIFO). 

\clearpage

\section{Appendix}
\subsection{The notation of Algorithm~\ref{alg:1} in detail}\label{app:alg_not}
Here, we want to manually calculate step $2$ of the Algorithm~\ref{alg:1} for $k\in\{12,3,1\}$ and $z=7$. This implies new data steps of $\lfloor z/k\rfloor \in\{0,2,7\}$. Let $k=3,u=3$. Thus, we want to compute the reduced updated third quarter of the hierarchy. Then, 
\begin{align*}
    \hat{ y}_{4(i-1)+3|z}^{[3](r)} = \hat{ y}_{4(i-1)+3|z}^{[3]} - \sum_{w=7}^7 { y}_{12(i-1)+w|z}^{[1]}.
\end{align*}
For $k=12,u=1$ we calculate the reduced annual forecast to be
\begin{align*}
    \hat y_{i|z}^{[12](r)} &= \hat y_{i|z}^{[12]} - \sum_{w=1}^2  y_{4(i-1)+w|z}^{[3]} - \sum_{w=7}^7  y_{12(i-1)+w|z}^{[1]} \\
    &= \hat y_{i|z}^{[12]} - \sum_{w=1}^7 \hat y_{12(i-1)+w|z}^{[1]},
\end{align*}
where both Eqs.~\eqref{eq:red},\eqref{eq:red_alt} were used.

\subsection{Proofs}\label{app:proof}
\begin{proof}[Proof of Theorem~\ref{th:general}]
    The covariance matrix of the updated and pruned base forecasts is given by $W_z=\text{Cov}(\mathbf y_z - \hat{\mathbf y}_z)$ of dimension $m_z\times m_z$. The pruned vectors are given by 
    \begin{align*}
        \mathbf y_z &= P_z\mathbf y-R_z\mathbf y\\
        \hat{\mathbf y}_z &= P_z\hat{\mathbf y}_{\cdot|z} - R_z\mathbf y,
    \end{align*}
    where $P_z$ is the $m_z\times m_0$ pruning matrix which is used to remove the corresponding nodes. $R_z$ is the reduction matrix used to adapt the pruned hierarchy as described in step $2$ of the procedure. Thus, we can write
    \begin{align*}
        W_z &= P_z\text{Cov}(\mathbf y-\hat{\mathbf y}_{\cdot|z})P_z' \\
            &= \check P_z W \check P_z',
    \end{align*}
    where $W=\text{Cov}(\mathbf y-\hat{\mathbf y})$ denotes the original base forecasts. We exploit the fact that the forecast horizons decrease after updates. Therefore, due to stationarity, we can use the original base forecast error covariance matrix and eliminate the corresponding entries using the $\check P_z$ matrix. The matrix $\check P_z$ also acts as a similar pruning matrix given by
    \begin{align*}
        \check P_z = \text{diag}(I_{m/k-\lfloor z/k\rfloor}~0_{m/k-\lfloor z/k\rfloor\times\lfloor z/k\rfloor},k\in K).
    \end{align*}

    Next, we can compare the covariance matrices of interest. We have 
    \begin{align*}
        P_z W P_z - \check P_z W \check Pz &= (P_z - \check P_z) W (P_z - \check P_z)' \\
        &= (W^{1/2}(P_z - \check P_z)')' (W^{1/2} (P_z - \check P_z)') \\
        &\geq 0.
    \end{align*}
    Hence, the difference of covariances matrices is positive semi-definite. Note that $\text{rk}(P_z-\check P_z)<m_z$. This also signifies that the updated base forecasts cannot independently improve forecasts at other levels. However, using the theory of the minT reconciliation method (see \cite{wickramasuriya2021propertiespointforecastreconciliation}) , we know that
    \begin{align*}
        \text{Cov}(\mathbf y_z - \hat{\mathbf y}_z) = W_z \geq \text{Cov}(\mathbf y_z - \tilde{\mathbf y}_z) = S_z(S_z'W_z^{-1}S_z)^{-1}.
    \end{align*}
    Altogether, this leads to
    \begin{align*}
        P_zWP_z' \geq W_z \geq S_z(S_z'W_z^{-1}S_z)^{-1},
    \end{align*}
    which implies that the diagonal elements, as well as the trace of the covariances matrices, are as the theorem states.
\end{proof}
\begin{proof}[Proof of Theorem~\ref{th:1}]
    The bottom level $\text{AR}(1)$ model leads to a top level model of $\text{ARMA}(1,1)$ according to Eq.~\eqref{eq:agg_arima}. Based on \cite{NEUBAUER2024BU}, the minT forecast reconciliation method is equivalent to the bottom-up approach. The second inequality is quickly proved using the arguments of \cite{Koreisha2004UpdatingAP} due to improvements in the top level forecasts using the bottom-up aggregated forecast. The bottom level forecast remains untouched here.

    For the first inequality, we take a look at the top level MSE, namely
    \begin{align}
        \text{MSE}_m(\tilde{\mathbf y}_0) &= \mathbb E\left[\left(y^{[m]}_i - \sum_{j=1}^{m}\hat y^{[1]}_{m(i-1)+j|m(i-1)}\right)^2\right] \nonumber \\
         &= \mathbb E\left[\left(\sum_{j=1}^{m} y^{[1]}_{m(i-1)+ j} - \hat y^{[1]}_{m(i-1)+j|m(i-1)}\right)^2\right] \nonumber \\
         &= \sigma^2 \mathbf 1_m' \Phi\Phi'\mathbf 1_m \label{eq:var_lemma} \\
         &\geq \sigma^2 (\mathbf 1_{m-z}'~\mathbf 0_z') \Phi\Phi'(\mathbf 1_{m-z}'~\mathbf 0_z')' \nonumber \\
         &= \mathbb E\left[\left(\sum_{j=z+1}^{m} y_{m(i-1)+j}^{[1]} - \hat y^{[1]}_{m(i-1)+j|m(i-1)+z}\right)^2\right] \label{eq:var_lemma_red}\\
         &= \mathbb E\left[\left(y^{[m]}_i - \sum_{j=1}^z y_{m(i-1)+j}^{[1]} - \sum_{j=z+1}^{m}\hat y^{[1]}_{m(i-1)+j|m(i-1)+z}\right)^2\right] \nonumber \\
         &= \text{MSE}_m(\tilde{\mathbf y}_z), \nonumber
    \end{align}

    where we applied \cite[Lemma 1]{NEUBAUER2024BU} in Eq.~\eqref{eq:var_lemma} and again in Eq.~\eqref{eq:var_lemma_red}. Similarly, at the bottom level the MSE can be reduced by    
    \begin{align*}
        \text{MSE}_1(\tilde{\mathbf y}_0) &= \sum_{j=1}^{m}\mathbb E\left[\left(y^{[1]}_{m(i-1)+j} - \hat y^{[1]}_{m(i-1)+j|m(i-1)}\right)^2\right] \\
        &= \frac{\sigma^2}{1-\phi^2}\sum_{j=1}^m (1-\phi^{2j}) \\
        &\geq \frac{\sigma^2}{1-\phi^2}\sum_{j=1}^{m-z} (1-\phi^{2j}) \\
        &= \sum_{j=z+1}^{m}\mathbb E\left[\left(y^{[1]}_{m(i-1)+j} - \hat y^{[1]}_{m(i-1)+j|m(i-1)+z}\right)^2\right] \\
        &= \text{MSE}_1(\tilde{\mathbf y}_z).
    \end{align*}
    Altogether, this concludes the proof.
\end{proof}
\clearpage
\subsection{Additional Plots}\label{app:plots}
\begin{figure}[!ht]
    \centering
    \includegraphics[width=1\textwidth]{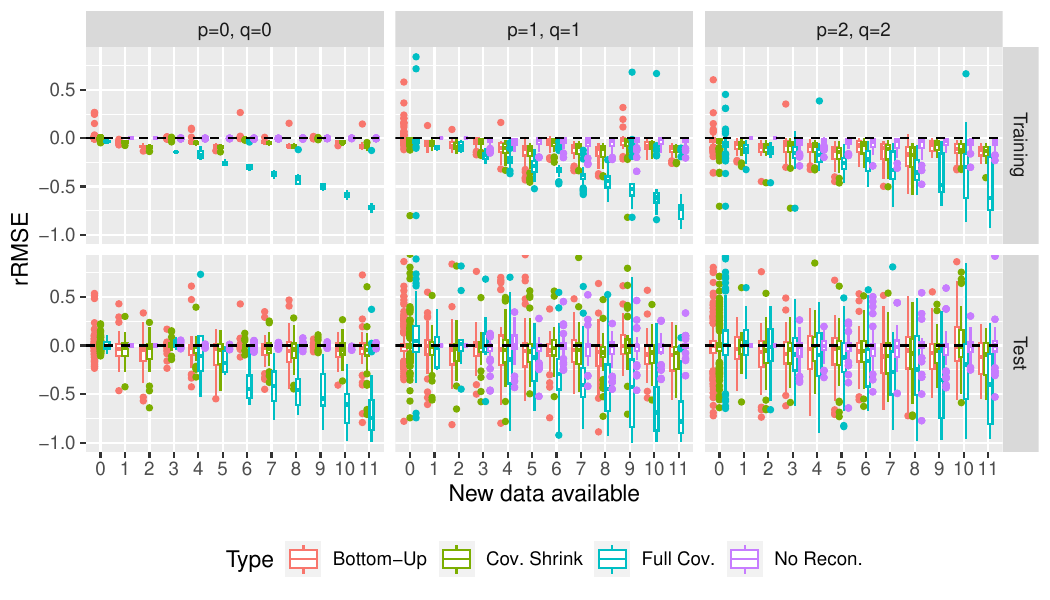}
    \caption{Training and Test rRMSE values of the middle level for $k\in\{12,3,1\}$, and various ARMA data-generating models and automatically selected fitted models.}
    \label{fig:mid_1241}
\end{figure}

\clearpage
\bibliographystyle{apalike}
\bibliography{refs}
\end{document}